\documentclass[DIV=classic,a4paper,10pt]{myart}

\KOMAoptions{DIV=last}
\AtBeginDocument{\AtBeginShipoutNext{\AtBeginShipoutDiscard}}
\usepackage{caption}
\usepackage{subcaption}

\renewcommand{\mid}{\,|\,}

\DeclareMathOperator*{\essinf}{ess\,inf}

\usepackage{accents}

\begin{document}

\title{Concentration of dynamic risk measures in a Brownian filtration}\thanks{The author thanks Julio Backhoff, Daniel Lacker and Peng Luo for helpful comments.}

\author[]{Ludovic Tangpi}

\abstract{
Motivated by liquidity risk in mathematical finance, D. Lacker \cite{Lac-concent} introduced concentration inequalities for risk measures, i.e. upper bounds on the \emph{liquidity risk profile} of a financial loss.
We derive these inequalities in the case of time-consistent dynamic risk measures when the filtration is assumed to carry a Brownian motion.
The theory of backward stochastic differential equations (BSDEs) and their dual formulation plays a crucial role in our analysis.
Natural by-products of concentration of risk measures are a description of the tail behavior of the financial loss and transport-type inequalities in terms of the generator of the BSDE, which in the present case can grow arbitrarily fast.
}

\date{\today}
%\ArXiV{91G80, 90B50, 60E10, 91B30}
\keyAMSClassification{60H20, 60H30, 91G20, 60E15.}
\keyWords{ Dynamic risk measures, backward stochastic differential equations, Brownian filtration, superquadratic growth, concentration inequalities, transportation inequalities.}
\maketitleludo
\setcounter{page}{1} %This is needed because of the mess created by the thank note that creates a new empty page.

\section{Introduction}

On a fixed probability space $(\Omega, {\cal F}, P)$, we investigate concentration inequalities of the form
\begin{equation}
\label{eq:concent}
	\rho(\lambda X) \le \lambda E[X]+ l(\lambda) \quad \text{ for all } \lambda \ge 0
\end{equation}
where $l$ is a given positive, increasing and convex function and $\rho:L^2\to (-\infty, +\infty]$ is a risk measure.
This type of inequalities first appeared in the work of \citet{Bob-Goet} and have recently been thoroughly studied by \citet{Lac-concent} in the context of \emph{liquidity risk} in finance.
Given a future uncertain loss $X$ the inequality \eqref{eq:concent} yields an upper bound on the liquidity risk profile $(\rho(\lambda X))_{\lambda \ge 0}$ of $X $. 
This (typically convex) curve captures the effect of leverage on risk.
Indeed, in a frictionless market, the minimal capital requirement  for (resp. the minimal price needed to hedge) a position $\lambda X$ should be linear in the scaling $\lambda$.
This is the case for coherent risk measures.
However, in an imperfect market, the minimal capital requirement (resp. superhedging price) is nonlinear in $\lambda$ and the slope of the curve at a given point (e.g. $\lambda\ge 1$) says how much an infinitesimal increase of the loss changes the minimal capital requirement. 
This slope is expected to be large for very illiquid claims.
The question is then for which risk measures and which claims can the liquidity risk profile be controlled and bounded? see \cite{Lac-concent}, where the author further gives an integrability condition under which \eqref{eq:concent} holds for shortfall risk measures.

The goal of this note is to derive the bound \eqref{eq:concent} for dynamic risk measures in the Brownian setting.
Given $T>0$, we equip $(\Omega,\mathcal{F},P)$ with the completed filtration $(\mathcal{F}_t)_{t \in [0,T]}$ of a $d$-dimensional Brownian motion $W$.
As usual, we identify random variables that are equal $P$-a.s. and understand equalities and inequalities in this sense.
A time-consistent dynamic risk measure is a family $(\rho_{s,t})_{0\le s\le t\le T}$ of convex increasing functionals $\rho_{s,t}:L^2({\cal F}_t)\to L^2({\cal F}_s)$ such that for all $0\le s\le t\le T$,
\begin{itemize}
\item[(A1)] $\rho_{s,t}(X+\eta) = \rho_{s,t}(X)+\eta$ for all $X \in L^2({\cal F}_t)$ and $\eta\in L^2({\cal F}_s)$
\item[(A2)] $\rho_{0,T}(0)=0$
\item[(A3)] $\rho_{r, t}(X) = \rho_{r,s}(\rho_{s,t}(X))$ for all $0\le r\le s\le t\le T$ and $X \in L^2({\cal F}_t)$
%\item[(A4)] $\rho_{s,t}(X_11_A + X_21_{A^c}) = \rho_{s,t}(X_1)1_A + \rho_{s,t}(X_2)1_{A^c}$ $P$-a.s. for all $X_1, X_2 \in L^2({\cal F}_t)$, $A\in {\cal F}_s$
\item[(A4)] for all $X_1, X_2 \in L^2({\cal F})$, if $X_1 \ge X_2$, then $X_1 = X_2$ if and only if $\rho_{0,T}(X_1)=\rho_{0,T}(X_2)$.
\end{itemize}
Note that convexity, (strict) monotonicity and the so-called monetary property (A1) are standard axioms of risk measure, see \cite{artzner01}.
The normalization condition (A2) is not mathematically essential, but is used for convenience purpose.
The time-consistency condition (A3) is of central importance.
It intuitively means that future decisions of a decision maker using the risk measures $(\rho_{s,t})_{0\le s\le t\le T}$ are consistent with the present ones.
We refer e.g. to  \citet{delbaen05}, \citet{foellmerpenner} and \citet{cheridito01} and the references therein for background on dynamic risk measures.
As shown by \citet{gianin02}, the $g$-expectations of \citet{peng03} constitute canonical examples of such dynamic risk measures.
The most prominent example being the entropic risk measure discussed in Example \ref{exa:entropic} below.

 Let us further consider the following property introduced by \citet{Coquet2002}: 
\begin{definition}
Given $\kappa>0$, we say that the risk measure $(\rho_{s,t})$ is $\kappa$-dominated  if for all $X_1, X_2 \in L^2$,
\begin{equation*}
	\rho_{0,T}(X_1 + X_2) - \rho_{0,T}(X_1) \le {\cal E}^\kappa(X_2)
\end{equation*}
 where ${\cal E}^\kappa(X_2):= Y_0$ and the pair $(Y, Z)$ is the unique square-integrable solution of the equation
\begin{equation*}
	-dY_t = \kappa  Z_t\,dt - Z_t\, dW_t, \quad Y_T = X_2.
\end{equation*}
\end{definition}

Let $C_0([0,T],\mathbb{R}^d)$ denote the space of continuous functions from $[0,T]$ to $\mathbb{R}^d$ starting from zero.
A function $\phi:C_0([0,T],\mathbb{R}^d)\to \mathbb{R}$ is said to be $1$-Lipschitz if
\begin{equation}
\label{eq:1 Lips on C}
	|\phi(w) - \phi(w')| \le \sum_{i=1}^d\sup_{t \in [0,T]}|w^i_t - w'^i_t|
\end{equation}
for all $w,w' \in C_0([0,T],\mathbb{R}^d)$.
% Let $\Lambda$ denote the $L^2$-closure of the set
% \begin{equation}
% \label{eq:setlambda}
% 	\left\{X=f(W_{T2^{-n}},\dots, W_{kT2^{-n}},\dots, W_T): n\in \mathbb{N},  f:\mathbb{R}^{d\times 2^{n}}\to\mathbb{R} \text{ is $1$-Lipschitz}  \right\},
% \end{equation}
% whereby $1$-Lipschitz means $|f(x)-f(y)|\le |x-y|$ for all $x,y\in \mathbb{R}^{d\times 2^{n}}$.
Our first main result is:
\begin{theorem}
\label{cor:dyn-RM}
	Let $(\rho_{s,t})$ be a time-consistent dynamic risk measure.
	If\footnote{$\rho^*_{0,T}$ denotes the convex conjugate of $\rho_{0,T}$ defined as $\rho^*_{0,T}(Z):=\sup_{X \in L^2}(E[ZX] - \rho_{0,T}(X))$, $Z \in L^2$.} $\rho^*_{0,T}(1)=0$ and there is $\kappa>0$ such that $(\rho_{s,t})$ is $\kappa$-dominated, then putting $l(\lambda):= \kappa(T\lambda^2 + 1)$ for $\lambda\ge0$, it holds
	\begin{equation}
	\label{eq:concen_RM}
		\rho_{0,T}(\lambda \phi(W)) \le \lambda E[\phi(W)] + l(\lambda)\quad \text{ for all } \phi:C_0([0,T],\mathbb{R}^d)\to \mathbb{R} \text{ 1-Lipschitz and } \lambda \ge 0.
	\end{equation}
\end{theorem} 
We also prove, see Theorem \ref{thm:main} below, the result for (minimal super-) solutions of backward stochastic differential equations (BSDEs) with quadratic growth.
These can be seen as dynamic risk measures which are not necessarily $\kappa$-dominated.
Furthermore, we prove \eqref{eq:concen_RM} even for BSDEs with generators growing arbitrarily fast is $z$, see Corollary \ref{cor:ForLipschitz}.

The proofs draw from convex duality results for BSDEs as developed by \citet{tarpodual} and the (finite dimensional) Gaussian transport-entropy inequality initiated by \citet{Talagrand96} 
and generalized to the framework of abstract Wiener space by \citet{Fey-Ust04}.
While we can hardly overstate the importance of these inequalities, see e.g. the survey of \citet{Goz-Leo10}, let us mention that \citet{Lac-concent} pointed out some interesting consequences of inequalities of the form \eqref{eq:concen_RM}: as mentioned above, they provide bounds on the liquidity risk profile of the financial derivative with payoff $X$ and, on the other hand, they give new types of transportation inequalities and new descriptions of the concentration of measure phenomenon. 
In the same vein, a dual formulation of the inequality \eqref{eq:concen_RM} leads us to a "transport-type inequality" in terms of the penalty function (convex conjugate) of the risk measure.
A particularly interesting case arises for risk measures stemming from BSDEs with superquadratic generators, as the penalty function is the integral of a subquadratic function. 

The next section is dedicated to the proof of Theorem \ref{thm:main}, where we further derive \eqref{eq:concen_RM} for minimal supersolutions of BSDEs. 
In the final section, we discuss some immediate mathematical consequences.
In particular, deviation and transport-type inequalities.

\section{Concentration of dynamic risk measures}
\paragraph{2.1 Minimal supersolutions of convex BSDEs.}
Consider the sets of processes 
\begin{equation*}\textstyle
	\mathcal{L}:=\left\{Z: \Omega\times [0,T] \to \mathbb{R}^d; Z\text{ is predictable, and }\int_{0}^{T}| Z_s |^2ds <+\infty \,P\text{-a.s.}\right\}
\end{equation*}
and $	\mathcal{S} :=\left\{ Y:\Omega \times [0,T] \to \mathbb{R} ; Y \text{ is adapted and c\`adl\`ag}\right\}$.
Following \citet{DHK1101}, a supersolution of the BSDE with terminal condition $X \in L^0$ and generator $g:\Omega\times [0,T]\times \mathbb{R}\times \mathbb{R}^d\to \mathbb{R}$ is defined as a pair $(Y,Z) \in \mathcal{S}\times \mathcal{L}$ such that %$\int ZdW$ is a supermartingale and 
\begin{equation}\textstyle
	\begin{cases}
 Y_s-\int_{s}^{t}g_u(Y_u,Z_u)du+\int_{s}^{t}Z_u dW_u\geq Y_t, \quad \text{for every} \quad 0\leq s\leq t\leq T\\
		\displaystyle Y_T\geq X.
	\end{cases}
	\label{eq:supersolutions}
\end{equation}
Given $X \in L^0$ and a generator $g$ define 
\begin{equation*}\textstyle
	\mathcal{A}\left(X\right):=\left\{(Y,Z)\in \mathcal{S}\times \mathcal{L}:(Y,Z)\text{ fulfills \eqref{eq:supersolutions} and $\int Z\,dW$ is a supermartingale} \right\}.
	\label{eq:set:supersolutions}
\end{equation*}
A supersolution $(\bar{Y},\bar{Z})\in \mathcal{A}(X)$ is said to be minimal if $\bar{Y}\leq Y$ for every $(Y,Z)\in \mathcal{A}(X)$.
It is shown in \cite{DHK1101} that if $X^-\in L^1$, the generator $g$ is $P\otimes dt$-a.s. positive, decreasing in $y$,  convex in $z$, and ${\cal A}(X)\neq \emptyset$, then \eqref{eq:supersolutions} admits a unique minimal supersolution $(\bar{Y}, \bar{Z})$.
In that case, it holds $\bar{Y}_t = \essinf\{Y_t: (Y,Z) \in \mathcal{A}(X)\}$ for all $t \in [0,T]$.
Our focus will be on the functional $\mathcal{E}: L^2 \to (-\infty, \infty] $ defined as
\begin{align*}
	\mathcal{E} :  X \longmapsto \begin{cases}{}
		\bar{Y}_0 &\text{ if } \mathcal{A}(X)\neq \emptyset\\
		+\infty & \text{ else}.
	\end{cases}
\end{align*}
This functional is convex, increasing and cash-subadditive, i.e. ${\cal E}(X+ \lambda) \le {\cal E}(X)+ \lambda$ for all $\lambda\ge0$, and if $g$ does not depend on $y$, then ${\cal E}(X+\lambda) = {\cal E}(X) +\lambda$, see \cite{DHK1101,tarpodual}.

\paragraph{2.3 Concentration inequalities.}
Further denoting by $\Lambda_+$ the  $L^2$-closure of 
\begin{equation*}
	\left\{X=f(W_{T2^{-n}},\dots, W_{kT2^{-n}},\dots, W_T): n\in \mathbb{N},  f:\mathbb{R}^{d\times 2^n}\to[0,\infty) \text{ is $1$-Lipschitz}  \right\}
\end{equation*}
whereby $1$-Lipschitz means $|f(x)-f(y)|\le |x-y|$ for all $x,y\in \mathbb{R}^{d\times 2^{n}}$,
one obtains:
\begin{theorem}
\label{thm:main}
	Assume that  $g$ is positive, decreasing in $y$ and jointly convex.
	Further assume that there are $a, b,c\ge 0$ such that
	\begin{equation}
	\label{eq:quad_growth}
		g_t(y,z)\le a + b|y| + c|z|^2 \quad \text{ for all }( y,z) \in \mathbb{R}\times \mathbb{R}^d.
	\end{equation}
	Then, putting $l(\lambda) := e^{bT}cT\lambda^2 + a$ for $\lambda\ge 0$, one has
	\begin{equation}
	\label{eq:concentration}
		{\cal E}(\lambda X) \le \lambda E[X] + l(\lambda)
		% 2b_2h(\frac{T\lambda}{2b_2}) + b_1
		\quad \text{ for all } \lambda \ge 0 \text{ and } X \in \Lambda_+.
	\end{equation}
	If in addition $g$ does not depend on $y$, then $\Lambda_+$ in \eqref{eq:concentration} can be replaced by $\Lambda$, the $L^2$-closure of the set
	\begin{equation*}
		\left\{X=f(W_{T2^{-n}},\dots, W_{kT2^{-n}},\dots, W_T): n\in \mathbb{N},  f:\mathbb{R}^{d\times 2^n}\to\mathbb{R} \text{ is $1$-Lipschitz}  \right\}.
	\end{equation*}
\end{theorem}

\begin{proof}
Consider the sets $\mathcal{D} :=\left\{\beta: \Omega\times [0,T] \to \mathbb{R}_+: \beta \text{ predictable and } \int_0^T\beta_u\,du < \infty \right\}$ and
\begin{align*}\textstyle
	 \mathcal{Q}:=\left\{q \in \mathcal{L}:\exp\left( \int_{0}^{T}q_u dW_u -\frac{1}{2}\int_{0}^{T}|q_u|^2 du \right)\in L^\infty\right\}.
\end{align*}
For $q \in \mathcal{Q}$, denote by $Q^q$ the probability measure whose density process is given by the stochastic exponential $M^q := \exp( \int_{}^{}q_u dW_u -\frac{1}{2}\int_{}^{}|q_u|^2 du )$ and for $\beta \in \mathcal{D}$, put $D^\beta_{s,t}:=\exp(-\int_s^t \beta_u du)$, $0\leq s\leq t\leq T$ and $D^\beta:= D^\beta_{0,T}$.

	It follows from \cite[Theorem 3.10]{tarpodual} that the functional ${\cal E}$ admits the convex dual representation
	\begin{equation}\textstyle
	\label{eq:rep_general}
		{\cal E}(X) =\sup_{(\beta, q)\in {\cal D}\times {\cal Q}} \left(E_{Q^q}\left[D^\beta X \right] - \alpha(\beta, q) \right), \quad X \in L^2
	\end{equation}
	with
	\begin{equation*}\textstyle
	 	\alpha(\beta, q) :=  E_{Q^q}\left[\int_0^TD^\beta_{0,u}g^*_u(\beta _u,q_u)\,du \right]
	\end{equation*}
	and $g^*$ the convex conjugate of $g$ given by $g^*(\beta, q):=\sup_{(y,z) \in\mathbb{R}\times \mathbb{R}^d}(-\beta y+ zq - g(y,z))$, $(\beta, q) \in \mathbb{R}_+\times \mathbb{R}^d$.
	When $g$ does not depend on $y$ the above do not depend on $\beta$ and take the form
	\begin{equation*}
			{\cal E}(X) =\sup_{ q\in  {\cal Q}} \left(E_{Q^q}\left[ X \right] - \alpha( q) \right) \quad\text{and} \quad \alpha( q)\textstyle :=  E_{Q^q}\left[\int_0^Tg^*_u(q_u)\,du \right].
	\end{equation*}
	It can be checked that due to \eqref{eq:quad_growth} and monotonicity of $g$, it holds $g^*(\beta, q)= \infty$ if $\beta<0$ or $\beta >b$, so that $\beta\in {\cal D}$ in the representation \eqref{eq:rep_general} can be assumed to satisfy $0\le \beta\le b$. 
	Let us assume that $c>0$.
	The growth condition \eqref{eq:quad_growth} on $g$ further implies that 
	\begin{equation}
	\label{eq:estim.g.star}
		g^*(\beta, q)\ge  \frac{1}{4c}|q|^2-a.
	\end{equation} 
	%with $\bar{c} =\frac{1}{4c}>0$.
	Observe that $D^{\beta}_{0,u}$ satisfies $e^{-bT}\le D^{\beta}_{0,u}\le 1$ for all $\beta$ such that $0\le \beta\le b$.
	Thus, multiplying both sides of \eqref{eq:estim.g.star} by $D^\beta_{0,u}$ and integrating yields
	% there is a positive constant again denoted $\bar{c}$ and depending only on $T,b,c$ such that
	$$\textstyle E_{Q^q}\left[\int_0^TD^\beta_{0,u}g^*_u(\beta_u,q_u)\,du \right] \ge \frac{1}{4c}e^{-bT}E_{Q^q}\left[\int_0^T|q_u|^2\,du \right]-a$$
	for all $(\beta, q) \in {\cal D}\times {\cal Q}$.
	Moreover, for each $q \in {\cal Q}$, 
	\begin{equation*}\textstyle
	 	\tau^n:=\inf\{t\ge 0 : |\int_0^tq_u\,dW_u|\ge n\}\wedge T
	 \end{equation*} 
	 defines a sequence of stopping times such that $\tau^n\uparrow T$ $P$-a.s.
	 Let $q^n:= q1_{[0,\tau^n]}$.
	 Then, the Kullback-Leibler information (or relative entropy) of $Q^q$ with respect to $P$ is given by
	 \begin{equation*}\textstyle
	H(Q^{q^n}|P):= E[M^{q^n}_T\log(M^{q^n}_T)] = E_{Q^{q^n}}\left[\int_0^Tq^n_u\,dW_u - \frac{1}{2}\int_0^T|q^n_u|^2\,du \right]
	       =\frac{1}{2}E_{Q^{q^n}}\left[\int_0^T|q^n_u|^2\,du\right],
	\end{equation*}
	where the last equality follows from Girsanov theorem.
	Since $M^q$ is continuous, $M^{q^n}_T= M^q_{\tau^n}\to M^q_T$ $P$-a.s. and by boundedness of $M^q$ and $|M^{q^n}_T\log(M^{q^n}_T)| \le (M^{q}_{\tau^n})^2 + 1\le E[(M^q_T)^2\mid{\cal F}_{\tau^n}]+1 $, it follows that $H(Q^{q^n}|P)\to H(Q^q|P)$.
	On the other hand, $M^{q^n}_T\int_0^T|q^n_u|^2\,du \le E[M^q_T\mid{\cal F}_{\tau^n}]\int_0^T|q_u|^2\,du$, showing that $E_{Q^{q^n}}\left[\int_0^T|q^n_u|^2\,du\right] \to E_{Q^{q}}\left[\int_0^T|q_u|^2\,du\right]$.
	 Thus, 
	\begin{equation*}\textstyle
	H(Q^q|P)
	       =\frac{1}{2}E_{Q^q}\left[\int_0^T|q_u|^2\,du\right],
	\end{equation*}
    and therefore
	\begin{equation}
	\label{eq:rhostar-entropy}
		\alpha(\beta, q) \ge \frac{1}{2c}e^{-bT}H(Q^q|P)-a.
	\end{equation}

	Given $q \in {\cal Q}$, define $M_n:= E[M^q\mid {\cal G}_n]$, where ${\cal G}_n$ is the $\sigma$-algebra given by ${\cal G}_n:= \sigma(W_{kT2^{-n}}: k = 0,\dots, 2^n)$, $n\in \mathbb{N}$.
	It holds $M_n = u_n(W_{T2^{-n}}, \dots, W_{kT2^{-n}}, \dots, W_T)$ for some bounded ${\cal B}(\mathbb{R}^{d\times 2^{n}})$-measurable function $u_n:\mathbb{R}^{d\times 2^{n}}\to \mathbb{R}_+$, where ${\cal B}(\mathbb{R}^{d\times 2^{n}})$ is the Borel $\sigma$-algebra of $\mathbb{R}^{d\times 2^{n}}$.  
	Let $q^n \in {\cal Q}$ be such that $M_n= M^{q^n}$ (this is justified by \cite[Proposition VIII.1.6]{Revuz1999}).
	It holds $M_n \to M^q$ $P$-a.s. and since %$\log$ is continuous and increasing and 
	$M^q$ is bounded, it follows by dominated convergence that $(M_n\log(M_n))$ converges in $L^1$, i.e. $H(Q^{q^n}|P) \to H(Q^q|P)$.
	For every $X \in L^2$, it follows by dominated convergence theorem, that $E_{Q^{q^n}}[D^\beta X] \to E_{Q^q}[D^\beta X]$.
	Combined with \eqref{eq:rhostar-entropy}, this shows that
	\begin{equation}
	\label{eq:rho-entropy}
		{\cal E}(X) \le \sup_{(\beta, q) \in {\cal D}\times\hat{\cal Q}}\left(E_{Q^q}\left[D^\beta X\right] - \frac{1}{2c}e^{-bT}H(Q^q|P) \right) + a,
	\end{equation}
	where $\hat{\cal Q}$ is the set of $q\in {\cal Q}$ such that $M^q = u(W_{T2^{-n}}, \dots, W_{kT2^{-n}}, \dots, W_T)$ for some $n\in \mathbb{N}$ and $u:\mathbb{R}^{d\times 2^{n}}\to \mathbb{R}_+$ bounded and Borel measurable.
	Hence, if $X$ is positive, as $D^\beta\le 1$ for all $\beta \in {\cal D}$, \eqref{eq:rho-entropy} yields
	\begin{equation}
	\label{eq:rho-qbeta-entropy}
				{\cal E}(X) \le \sup_{ q \in \hat{\cal Q}}\left(E_{Q^q}\left[X\right] - \frac{1}{2c}e^{-bT}H(Q^q|P) \right) + a
	\end{equation}
	and if $g$ does not depend on $y$, \eqref{eq:rho-qbeta-entropy} holds for all $X \in L^2$.
	For $q \in \hat{\cal Q}$ and $u:\mathbb{R}^{d\times 2^{n}}\to \mathbb{R}_+$ a bounded Borel measurable function such that $M^q = u(W_{T2^{-n}}, \dots, W_{kT2^{-n}}, \dots, W_T)$, we denote by $\gamma^{dn}$ the canonical Gaussian measure on $\mathbb{R}^{d\times 2^{n}}$ and by $\mu_{q}$ the Borel measure on $\mathbb{R}^{d\times 2^{n}}$ given by
	\begin{equation}
	\label{eq:muq}
		\frac{d\mu_{q}}{d\gamma^{dn}}:=u(\sqrt{T2^{-n}}x_1,\dots, \sqrt{kT2^{-n}}x_k,\dots,\sqrt{T}x_n), \quad x_i\in \mathbb{R}^d.
 	\end{equation}
  	The independence of the increments of $W$, yields
	\begin{equation}
	\label{eq:entro_Q-to-mu}
		E[M^q\log(M^q)] = \int_{\mathbb{R}^{d\times 2^n}}\frac{d\mu_{q}}{d\gamma^{dn}}\log(\frac{d\mu_{q}}{d\gamma^{dn}})\,d\gamma^{dn}=:H(\mu_{q}|\gamma^{dn}).
	\end{equation}
	That is, $H(Q^q|P) = H(\mu_{q}|\gamma^{dn})$, where the right hand side denotes the Kullback-Leibler information of $\mu_{q}$ with respect to $\gamma^{dn}$ on the measurable space $(\mathbb{R}^{d\times 2^{n}},{\cal B}(\mathbb{R}^{d\times 2^{n}}))$.
	Therefore, by Talagrand's $T_1$ inequality, see \cite[Theorem 1.1]{Talagrand96}, one has
	\begin{equation}
	\label{eq:Talagrand}
		h(\mathcal{W}_{1,dn}(\gamma^{dn},\mu_{q}))\le H(\gamma^{dn}|\mu_{q})
 	\end{equation}
 	with $h(x)=x^2/2$ and $\mathcal{W}_{1,dn}(\gamma^{dn},\mu_{q})$ denoting the Wasserstein distance of order $1$ between $\gamma^{dn}$ and $\mu_{q}$ given by
	\begin{equation*}\textstyle
		\mathcal{W}_{1,dn}(\gamma^{dn}, \mu_{q}) := \inf\iint|x-y|\,d\pi(x,y) 
	\end{equation*}
	where the infimum is over the set of probability measures $\pi$ on $\mathbb{R}^{d\times 2^{n}}\times \mathbb{R}^{d\times 2^{n}}$ with marginals $\gamma^{dn}$ and $\mu_{q}$.
	Given $n \in \mathbb{N}$ and a $1$-Lipschitz function $f:\mathbb{R}^{d\times 2^n}\to \mathbb{R}_+$  consider the random variable $X:=f(W_{T2^{-n}}, \dots, W_{kT2^{-n}}, \dots, W_T)$. 
	It is clear that $X \in L^2$ %follows from Doob's maximal inequality that $X \in {\cal H}^1$ 
	and so, ${\cal E}(X)$ satisfies \eqref{eq:rho-qbeta-entropy}.	
	Let $q\in \hat{\cal Q}$ and $\mu_{q}$ be given by \eqref{eq:muq}.
	Since the function $\tilde{f}(x_1,\dots, x_n):= \frac{1}{\sqrt{T}}f(\sqrt{T2^{-n}}x_1,\dots, \sqrt{kT2^{-n}}x_k,\dots,\sqrt{T}x_n)$ is again $1$-Lipschitz, by the Kantorovich-Rubinstein formula, see e.g. \cite[Section 5]{Vil2} the inequality \eqref{eq:Talagrand} implies that $h(\int \tilde{f} \,d\mu_{q} - \int \tilde{f}\,d\gamma^{dn}) \le H(\gamma^{dn}|\mu_{q})$.
	Thus, for all $\lambda\ge 0 $ and by definition of the convex conjugate $h^*$ of $h$, one has
	\begin{equation*}\textstyle %\frac{\sqrt{T}\lambda}{2\bar{c}}
		2\lambda ce^{bT}\sqrt{T}\left(\int \tilde{f} \,d\mu_{q} - \int \tilde{f}\,d\gamma^{dn} \right) - h^*(2\lambda ce^{bT}\sqrt{T}) \le H(\gamma^{dn}|\mu_{q}),
	\end{equation*}
	so that arguing as in the computations leading to \eqref{eq:entro_Q-to-mu}, one obtains
	\begin{equation*}
		E_{Q^q}[\lambda  X] - \frac{1}{2c}e^{-bT}H( Q^q|P)\le E[\lambda X] +\frac{1}{2c}e^{-bT}h^*(2\lambda ce^{bT}\sqrt{T}).
	\end{equation*}
	Therefore, due to \eqref{eq:rho-qbeta-entropy}, taking the supremum above as $q$ runs over $\hat{\cal Q}$ yields
	\begin{equation*}
		{\cal E}(\lambda X) \le E[\lambda X] +\frac{1}{2c}e^{-bT}h^*(2\lambda ce^{bT}\sqrt{T}) + a =: E[\lambda X] + l(\lambda).
	\end{equation*}
	Since $h(x) = x^2/2$, it holds $l(\lambda) = e^{bT}cT\lambda^2 + a$.

	Now, let $X \in \Lambda_+$.
	There is a sequence $X^n$ of the form $X^n = f^n(W_{T2^{-n}}, \dots, W_{kT2^{-n}}, \dots, W_T)$ where $f^n$ is a positive and $1$-Lipschitz function on $\mathbb{R}^{d\times 2^{n}}$ which converges to $X$ in $L^2$, and therefore in $L^1$.
	As above, for all $n$ and $\lambda \ge 0$ one has
	\begin{equation*}
		{\cal E}(\lambda X^n)  \le E[\lambda X^n] + l(\lambda).
	\end{equation*}
	Since ${\cal E}$ is $L^1$-lower semicontinuous, see \cite[Theorem 4.9]{DHK1101}, taking the limit as $n$ goes to infinity, one has
	\begin{equation*}
		{\cal E}(\lambda X)  \le E[\lambda X] +l(\lambda).
	\end{equation*}
	This yields \eqref{eq:concentration}.

	Recall that working with positive random variables was necessary to obtain \eqref{eq:rho-qbeta-entropy}.
	Since this inequality holds for all $X\in L^2$ when $g$ does not depend on $y$, the last claim of the theorem follows by the same arguments as above.

	In case $c=0$, it can be assumed w.l.o.g. that $g$ does not depend on $z$, so that the functional ${\cal E}$ satisfies the representation
	\begin{equation*}\textstyle
		{\cal E}(X) = \sup_{\beta\in {\cal D}}E\left[D^\beta X-\int_0^TD^\beta_ug^*_u(\beta_u)\,du\right]
	\end{equation*}
	where due to the growth condition \eqref{eq:quad_growth}, the supremum can be restricted to $\beta\in {\cal D}$ satisfying $0\le \beta\le b$.
	This growth condition further ensures that $g^\ast(\beta)\ge-a$.
	Thus, for all $X \in L^2$ with $X\ge 0$, one has
	\begin{equation*}
		{\cal E}(\lambda X)\le E[\lambda X] + a\quad \text{for every } \lambda\ge0.
	\end{equation*}
	This concludes the proof.
\end{proof}

\begin{example}
\label{exa:entropic}
	A particularly interesting example arises when $g_t(y,z)=|z|^2/2$. 
	In this case, ${\cal E}(X)= \log(E[e^X])$ for every $X$ with finite exponential moment and ${\cal E}$ is not ${\cal E}^\kappa$-dominated for any $\kappa >0$.
	If moreover $\Omega=C([0,T],\mathbb{R}^d)$ is the canonical space equipped with the uniform distance and $P$ the Wiener measure, then by \cite[Theorem 1.3]{Bob-Goet}, the Gaussian transportation inequality of \citet{Fey-Ust04} implies the bound ${\cal E}(\lambda X) \le \lambda E[X] + \lambda^2/2$ for every $\lambda\ge0$ and every $1$-Lipschitz function $X:\Omega\to \mathbb{R}$.
\end{example}

Two questions emerging from Theorem \ref{thm:main} and its proof are whether the set $\Lambda$ of "claims" for which the concentration inequality \eqref{eq:concentration} holds can be enlarged, and whether Theorem \ref{thm:main} can be obtained for generators growing faster than quadratic (in $z$).
Regarding the first question, one possibility could be to use the Kantorovich duality, (see e.g. \citet[Theorem 5.10]{Vil2}) instead of the Kantorovich-Rubinstein formula in the proof of Theorem \ref{thm:main}. 
This, however, would destroy the interpretation of \eqref{eq:concentration} as a bound on the liquidity risk profile.
To answer the second question, one can consider a slight modification of the set $\Lambda$.

Denote by $\Delta W_{jT2^{-n}}:= W_{jT2^{-n}} - W_{(j-1)T2^{-n}}$ the increment of the Brownian motion, by $\Lambda'$ the $L^2$-closure of the set
\begin{equation*}\textstyle
 \{X:=f(\Delta W_{T2^{-n}},\Delta W_{2T2^{-n}},\dots, \Delta W_{T}); n\in \mathbb{N},\,\, f:\mathbb{R}^{d\times 2^{n}}\to \mathbb{R} \text{ is $1$-Lipschitz}   \}
\end{equation*}
and $\Lambda_+'$ the $L^2$-closure of
\begin{equation*}\textstyle
 \{X:=f(\Delta W_{T2^{-n}},\Delta W_{2T2^{-n}},\dots, \Delta W_{T}); n\in \mathbb{N},\,\, f:\mathbb{R}^{d\times 2^{n}}\to [0,\infty) \text{ is $1$-Lipschitz}   \}.
\end{equation*}
Using the fact that the elements of $\Lambda'$ have uniformly bounded Malliavin derivative, one obtains the next corollary.
Refer to \citet{Nualart2006} for elements of Malliavin calculus.
\begin{proposition}
\label{prop:fastgrowth}
	Assume that $g$ is positive, decreasing in $y$ and jointly convex.
	Further assume that $g_t(y,0)=0$ and there is $b\ge 0$ and non-decreasing functions $\varphi, \theta:[0,\infty) \to [0,\infty)$ such that
	\begin{equation}\textstyle
	\label{eq:fastgrowth}
		|g_t(y, z) - g_t(y', z')|\le b|y-y'| + \varphi(|z|\vee |z'|)|z - z'| \text{ and } |\varphi(z)-\varphi(z')|\le \theta(|z|\vee|z'|)|z-z'|
	\end{equation}
	for all $y,y' \in \mathbb{R}$ and $z,z'\in \mathbb{R}^d$.
	Then, putting $c := \varphi(Te^T)+Te^T\theta(Te^T)$ and $l(\lambda) := e^{bT}cT\lambda^2$ for $\lambda\ge 0$ one has
	\begin{equation}\textstyle
	\label{eq:concent_superqu}
		{\cal E}(\lambda X) \le E[\lambda X] + l(\lambda) \quad \text{for all } \lambda \ge 0\text { and } X \in \Lambda'_+.
	\end{equation}
	If $g$ does not depend on $y$, then $\Lambda'_+$ in \eqref{eq:concent_superqu} can be replaced by $\Lambda'$.
\end{proposition}
\begin{proof}
Each $X$ of the form $X:= f(\Delta W_{T2^{-n}},\Delta W_{2T2^{-n}},\dots, \Delta W_{T})$ for some $n$ and $f:\mathbb{R}^{d\times 2^{n}}\to \mathbb{R}$ $1$-Lipschitz has bounded Malliavin derivative.
In fact, by \cite[Proposition 1.2.4]{Nualart2006}, it holds
\begin{equation*}\textstyle
|D_tX| \le \sum_{j=1}^{2^{n}}|D_t\Delta W_{jT2^{-n}}| \le 1.
\end{equation*}
Consider the function $\tilde{g}_t(y,z):= b|y| +\varphi(|z|)|z|$.
This (deterministic) function has constant Malliavin derivative and satisfies
\begin{equation*}\textstyle
	|\tilde{g}_t( y, z) - \tilde{g}_t(y', z')|\le b|y-y'| + \Big(\varphi(|z|) + |z'|\theta(|z|\vee|z'|)\Big)|z - z'| 
\end{equation*}
for all $y,y' \in \mathbb{R}$ and $z,z'\in \mathbb{R}^d$.
Thus, by \cite[Theorem 2.2 and Remark 2.3]{Che-Nam}, the BSDE with terminal condition $X$ and generator $\tilde{g}$ admits a unique solution $(Y,Z)$ with $Y\in {\cal S}^2$ and $|Z_t| \le Te^{T}=:K$. 
Hence, $(Y,Z)$ coincides with the unique solution of the BSDE with terminal condition $X$ and the Lipschitz continuous generator
\begin{equation*}\textstyle
	\hat{g}_t(y, z):= \begin{cases} \tilde{g}_t(y, z)& \text{if } |z| \le K\\
	\tilde{g}_t(y, Kz/|z|) &\text{if } |z| > K
	 \end{cases}
\end{equation*}
(which is known to exist by e.g. \cite{karoui01}).
Notice that the bound on $Z$ does not depend on $n$, so that $\hat g$ does not depend on $n$ as well.
On the one hand, $g \le \tilde g$, and thus, from \cite[Proposition 3.3]{DHK1101}, ${\cal E}(X) \le \tilde {\cal E}(X)$, where $\tilde {\cal E}(X)$ is the minimal supersolution of the BSDE with generator $\tilde g$.
Moreover, $\hat{g}$ satisfies the conditions of Theorem \ref{thm:main} and $Y_0$ coincides with $\hat{\cal E}(X)$, the minimal supersolution of the BSDE with terminal condition $X$ and generator $\hat g$.
Thus, 
\begin{equation*}\textstyle
	{\cal E}(X)\le \tilde{\cal E}(X) \le Y_0 = \hat{\cal E}(X) = \sup_{(\beta, q)\in {\cal D}\times {\cal Q}}E_{Q^q}\left[D^\beta X - \int_0^TD^\beta_{0,u}\hat{g}_u^*(\beta_u, q_u)\,du \right]
\end{equation*}
where $\hat{g}^*$ is the convex conjugate of $\hat g$.
Since $\hat{g}$ satisfies the growth condition \eqref{eq:quad_growth} with $a = 0$ and $c = \varphi(Te^T)+Te^T\theta(Te^T)$, as in the proof of Theorem \ref{thm:main}, one has
\begin{equation*}\textstyle
	{\cal E}(X) \le   \sup_{(\beta, q)\in {\cal D}\times {\cal Q}}\left( E_{Q^q}\left[D^\beta X\right] - \frac{1}{2c}e^{-bT}H(Q^q|P) \right).
\end{equation*}
In particular, a discretization and approximation argument already described puts us in the situation of Equation \eqref{eq:rho-entropy}.
The rest of the proof follows exactly as in the proof of Theorem \ref{thm:main}.
\end{proof}
\begin{corollary}
\label{cor:ForLipschitz}
	Assume that $g:[0,T]\times \Omega\times \mathbb{R}\times \mathbb{R}^d\to\mathbb{R}$ and $l:\mathbb{R}_+\to \mathbb{R}_+$ are as in Proposition \ref{prop:fastgrowth}.
	Then, it holds
	\begin{equation*}
		{\cal E}(\lambda \phi(W))\le \lambda E[\phi(W)] + l(\lambda)
	\end{equation*}
	for all $\lambda \ge 0$ and $\phi:C_0([0,T],\mathbb{R}^d)\to \mathbb{R}_+$ $1$-Lipschitz.
	If $g$ does not depend on $y$, then the result also holds for $\phi:C_0([0,T],\mathbb{R}^d)\to \mathbb{R}$ $1$-Lipschitz.
\end{corollary}
\begin{proof}
	In view of Proposition \ref{prop:fastgrowth}, it suffices to show that for every $\phi:C_0([0,T],\mathbb{R}^d)\to \mathbb{R}$ $1$-Lipschitz, $\phi(W) \in \Lambda'$.
	This in turn follows by arguments in \cite{Che-Nam} that we sketch below.
	The reader is referred to the proof of \cite[Proposition 3.2]{Che-Nam} for details.
	Put $t^n_k:=kT2^{-n}$, with $k=0,\dots,2^n$ and define the function $f^n:\{x=(x_k)_{k=1,\dots,2^n}: x_k \in \mathbb{R}^d \}\to C_0([0,T],\mathbb{R}^d)$ by
	\begin{equation*}
		f_0^n(x):=0,\quad f^n_t(x):= x_1 + \cdots+ x_{k-1} + \frac{t - t^n_{k-1}}{T2^{-n}}x_k,\quad t \in (t^n_{k-1}, t^n_k].
	\end{equation*}
	Further denote $X^n:= \phi\circ f^n(\Delta W_{t^n_1},\dots, \Delta W_{t^n_{2^n}})$.
	The function $\phi\circ f^n:\mathbb{R}^{d\times 2^n}\to \mathbb{R}$ is $1$-Lipschitz.
	In fact,
	\begin{align*}
		|\phi \circ f^n(x) - \phi\circ f^n(y)| &\le \sup_{t \in [0,T]}|f^n_t(x) - f^n_t(y)|\\
							&\le |x - y|
	\end{align*}
	for every $x,y\in \mathbb{R}^{d\times 2^n}$.
	Moreover, using the fact that Brownian motion has stationary increments together with Doob maximal inequality, it follows that for every $p\ge 2$ there is a constant $c_p$ depending only on $p$ such that
	\begin{equation*}
		E[|\phi(W) - X^n|^p] \le c_p\sqrt{T^p}(2^n)^{1-p/2}.
	\end{equation*}
	This shows that $(X^n)_n$ converges to $\phi(W)$ in $L^p$ for every $p>2$.
	Thus, $(X^n)_n$ converges to $\phi(W)$ in $L^2$ and therefore $\phi(W) \in \Lambda'$.
\end{proof}

\begin{proof}(of Theorem \ref{cor:dyn-RM})
	If $\rho_{0,T}$ is ${\cal E}^\kappa$-dominated for some $\kappa>0$ then it follows from \cite[Theorem 7.1]{Coquet2002} that there is a unique function $g:[0,T]\times \Omega\times \mathbb{R}^d\to \mathbb{R}$ such that for all $X \in L^2$, $\rho_{t,T}(X) = Y_t$, and $(Y,Z)$ is the unique solution of the BSDE
	\begin{equation}
	\label{eq:bsde_lip}
		-dY_t = g_t(Z_t)\,dt - Z_t\, dW_t, \quad Y_T = X.
	\end{equation}
	Moreover, $g$ satisfies $g(0) = 0$,
	\begin{equation}
	\label{eq:lipschtz_g}
		|g_t(z_1) - g_t(z_2)| \le C|z_1-z_2|
	\end{equation}
	for some $C\in (0,\kappa]$ and for all $z_1,z_2 \in \mathbb{R}^d$, $g(z)$ predictable and $E[\int_0^T|g_s(z)|^2\,ds]<\infty$ for all $z \in \mathbb{R}^d$.
	In particular, $g$ satisfies the condition \eqref{eq:fastgrowth} with $\varphi \equiv \kappa$, $\theta\equiv 0$ and $b = 0$.
	Since $\rho_{t,T}$ is convex, it follows from \cite[Theorem 3.2]{Jiang} that the function $g$ is convex, and using convex duality arguments,
	one has
	\begin{equation*}\textstyle
		\rho_{0,T}(X) \ge \sup_{Z\in L^2}(E[ZX] - \rho^*_{0,T}(Z)),\quad X \in L^2.
	\end{equation*}
	In particular, $\rho_{0,T}(X)\ge E[X]-\rho^*_{0,T}(1) = E[X]$.
	Thus, by the converse comparison for BSDEs, see e.g. \cite[Theorem 4.4]{Bri-Coq-Hu-Mem-Peng} combined with \cite[Lemma 2.1]{Jiang}, it holds $g_t(z)\ge 0$.
	Hence, $g$ is positive and convex.
	Furthermore, since $g$ satisfies \eqref{eq:lipschtz_g}, the unique solution $Y$ of \eqref{eq:bsde_lip} coincides with the minimal supersolution for all $X \in L^2$, see \cite[Remark 3.6]{tarpodual}.
	Thus, the result follows from Corollary \ref{cor:ForLipschitz}.
\end{proof}
The following is a natural example of ${\cal E}^\kappa$-dominated time-consistent dynamic risk measure.
\begin{example}
	Let $g:[0,T]\times \Omega\times \mathbb{R}^d\to \mathbb{R}$ be convex and Lipschitz continuous in the third variable with Lipschitz constant $\kappa\ge0$.
	Furthermore, assume $g_t(0)=0$ and $E[\int_0^T|g_t(z)|^2\,dt]<\infty$ for all $z \in \mathbb{R}^d$.
	Then, the dynamic risk measure $\rho_{s,t}(X) := Y_s$, where $(Y,Z)$ is the unique solution of the BSDE
	\begin{equation*}
		-dY_u = g_u(Z_u)\,du - Z_u\, dW_u, \quad Y_t = X \in L^2({\cal F}_t)
	\end{equation*}
	is time-consistent and ${\cal E}^\kappa$-dominated, see \cite{Coquet2002}.
\end{example}

\begin{example}
	Random variables $X \in \Lambda$ can be seen as claims, or financial losses whose liquidity risks is to be evaluated.
	For instance, consider a model with a stock following the dynamics $dS_t = S_t(b_t\,dt + \sigma_t\,dW_t)$ for some $b \in L^1(dt)$ and $\sigma^{j}:[0,T]\to [-1,1]$ measurable for all $j=1,\dots,d$.
	Let $X:= F(S_T)$ be the payoff of a contingent claim written on $S$.
	If $F(x) = \log(x)$, i.e. $X$ is the payoff of the so-called log-contract, see e.g. \cite{Neub94}, then
	\begin{equation*}\textstyle
		X= \log(S_0)+ \int_0^T(b_t -\frac{\sigma_t}{2})\,dt + \int_0^T\sigma_t\,dW_t.
	\end{equation*}
	It is clear that the r.h.s. above is in $\Lambda$, as $L^2$-limit of discrete-times stochastic integrals and since $|\sigma_t|\le 1$.
	Thus, for every time-consistent dynamic risk measures $(\rho_{s,t})$, that is ${\cal E}^\kappa$-dominated for some $\kappa>0$ there is a convex increasing function $l$ such that $\rho_{0,T}(\lambda X) \le \lambda E[X] + l(\lambda)$ for all $\lambda\ge 0$.
\end{example}
As mentioned in the introduction, there is a strong relation between BSDEs and time-consistent dynamic risk measures.
On the one hand, BSDEs allow to define dynamic risk measures.
This was first studied in \cite{gianin02,Jiang,peng02} for BSDEs with Lipschitz continuous generators and extended to the quadratic growth case in \cite{Hu-Ma-Peng-Yao,Kazi.Tani16} and to convex BSDEs in \cite{DHK1101,tarpodual}.
Reciprocally, under suitable additional conditions, time-consistent dynamic monetary risk measures can be seen as solutions of BSDEs. 
This "dynamic representation" of risk measures was initiated by \cite{Coquet2002} who introduced the concept of ${\cal E}^\kappa$-domination.
Refer to \cite{Hu-Ma-Peng-Yao,Del-Peng-Ros} and the references therein for further development.
\begin{remark}
	In light of Theorems \ref{thm:main} and \ref{cor:dyn-RM} it is fair to wonder whether concentration bounds for time consistent risk measures under weaker domination conditions as ${\cal E}^\kappa$-domination can be derived.
	The essential argument for such an extension is a BSDE representation of the risk measure.
	This can be obtained under a weaker domination assumption, but with additional structural conditions on the risk measure. See e.g. \cite[Theorem 6.3]{Hu-Ma-Peng-Yao} for details.
	Furthermore, the work of \citet{Del-Peng-Ros} allows to derive a BSDE representation of dynamic risk measures without assuming any domination condition.
	Concentration inequalities for such risk measures is left for future research.
\end{remark}
Let us mention that \citet{Lac-concent} gave an integrability criterion for the concentration property of static risk measures including optimized certainty equivalent and shortfall risk measures.
In the case where $\Omega$ is equipped with a metric $\delta$, \citet{Bob-Ding15} gave integral criteria on $\delta$ from which the concentration property for optimized certainty equivalent risk measures with power-type loss functions can be derived.
Notice that these functionals are typically time-inconsistent, except for the entropic case.

\section{Applications}
\label{sec:appli}

\paragraph{3.1 Transport-type inequalities.}
Let us borrow an argument from \citet{Goz-Leo10} to show that Theorem \ref{thm:main} leads to a transport-type inequality. 
In what follows, we put
$$\mathcal{W}_1( Q^q,P):= \sup_{X \in \Lambda}(E_{Q^q}[X]-E_P[X])\quad \text{and} \quad \mathcal{W}_1'( Q^q,P):= \sup_{X \in \Lambda'}(E_{Q^q}[X]-E_P[X]).$$
\begin{proposition}
\label{prop:transp-bsde}
	Let $g:\Omega\times [0,T]\times\mathbb{R}^d\to [0,\infty)$ be convex in $z$.
	If $g$ satisfies \eqref{eq:quad_growth}, then one has
	\begin{equation}\textstyle
	\label{eq:transport-bsde}
		l^*(\mathcal{W}_1( Q^q, P))\le E_{Q^q}[\int_0^Tg^*_u(q_u)\,du] \quad \text{for all } q \in  {\cal L},
	\end{equation}
	with $l$ as in Theorem \ref{thm:main}.
	If $g$ satisfies \eqref{eq:fastgrowth}, then one has
	\begin{equation}\textstyle
	\label{eq:transport-bsdefast}
		l^*(\mathcal{W}_1'( Q^q, P))\le E_{Q^q}[\int_0^Tg^*_u(q_u)\,du] \quad \text{for all } q \in  {\cal L},
	\end{equation}
	with $l$ as in Proposition \ref{prop:fastgrowth}.
\end{proposition}
\begin{proof}
It follows from Theorem \ref{thm:main} and the convex dual representation of BSDEs (see \cite[Theorem 3.10]{tarpodual}) that for all $X \in \Lambda$, $\lambda\ge0$ and $ q \in {\cal L}$,
\begin{equation*}\textstyle
	\lambda E_{Q^q}[ X] - E_{Q^q}[\int_0^Tg^*_u(q_u)\,du] \le {\cal E}(\lambda X)\le \lambda E[X] + l(\lambda).
\end{equation*}
Therefore, $\lambda(E_{Q^q}[ X] - E[X]) - l(\lambda)\le E_{Q^q}[\int_0^Tg^*_u(q_u)\,du]$. Taking the supremum over $X \in \Lambda$ and $\lambda\ge 0$ yields \eqref{eq:transport-bsde}.
The proof of \eqref{eq:transport-bsdefast} is the same.
\end{proof}

\begin{remark}
The transport-type inequality \eqref{eq:transport-bsde} can also be written w.r.t. subprobability measures, i.e. measures with density $D^\beta Q^q$ for some $(\beta, q)\in {\cal D}\times {\cal Q}$.
In fact, the arguments of Proposition \ref{prop:transp-bsde} lead to
\begin{equation*}
			l^*(\mathcal{W}^+_1(D^\beta Q^q, P))\le \alpha(\beta,q)\quad \text{for all } q \in  {\cal L},
\end{equation*}
with $\mathcal{W}^+_1(D^\beta Q^q,P):= \sup_{X \in \Lambda_+}(E_{Q^q}[D^\beta X]-E_P[X])$.
\end{remark}

\begin{corollary}
	Assume that $\Omega= C([0,1],\mathbb{R}^d)$ is the canonical space, $W$ the canonical process and $P$ the Wiener measure.
	Equip $\Omega$ with the augmented natural filtration of $W$.
	Further put $||\omega||_1:= \sum_{i = 1}^d\sup_{t \in [0,T]}|\omega_t^i|$ and
	%  $||\omega||_1:= \int_0^T|\omega_t|\,dt$ be the $L^1(dt)$-norm on $\Omega$ and
	% \begin{equation*}\textstyle
	% 	\mathcal{W}_1(Q_1,Q_2):= \inf_{\pi}\iint ||\omega_1 -\omega_2 ||_1\,d\pi(\omega_1,\omega_2),
	% \end{equation*}
	\begin{equation*}\textstyle
		\mathcal{W}_1(Q_1,Q_2):= \inf_{\pi}\iint ||\omega_1 -\omega_2 ||\,d\pi(\omega_1,\omega_2),
	\end{equation*}
	where the infimum is over probability measures on $\Omega\times \Omega$ with marginals $Q_1$ and $Q_2$.
	Then, for $g$ and $l$ as in Proposition \ref{prop:fastgrowth}, it holds
	\begin{equation}\textstyle
	\label{eq:transport-bsde_cano}
		l^*\left( \mathcal{W}_1(Q^q,P) \right) \le E_{Q^q}[\int_0^Tg^*_u(q_u)\,du]\quad \text{for all} \quad q \in {\cal L}.
	\end{equation}
\end{corollary}
\begin{proof}
	% Let $X:\Omega \to\mathbb{R}$ be a $1$-Lipschitz function.
	% Let $\varepsilon>0$ by the Kantorovich-Rubinstein formula, there is $X:\Omega \to \mathbb{R}$ such that
	% \begin{equation}\textstyle
	% \label{eq:trans.rubins}
	% 	\mathcal{W}_1(Q^q,P)\le E_{Q^q}[X]-E[X] + \varepsilon.
	% \end{equation}
	% For every $\omega \in \Omega$, there is a sequence of discrete paths $\omega^n:=\sum_{i=1}^{N^n}c^n_i1_{[t^n_i,t^n_{i+1})}$ converging pointwise to $\omega$.
	% Let 
	% $X^n(\omega):=X(\omega^n)$. 
	% With some abuse of notation, $ X^n\in \Lambda$.
	% Let $q \in {\cal L}$.
	% By Lipschitz continuity of $X$, it holds $X^n\to X$ $P$-a.s. and $Q^q$-a.s., and for all $n\in \mathbb{N}$, $|X^n(\omega)|=|X(\omega^n)|\le ||\omega^n||_1 + X(0)\le \sup_{t\in [0,T]}|\omega_t| + X(0)$, showing that $X^n$ is uniformly integrable w.r.t. $P$ and $Q^q$.
	% Thus, $E_{Q^q}[X]-E[X] = \lim_{n}E_{Q^q}[X^n]-E[X^n]$.
	% Thus, it follows from \eqref{eq:trans.rubins} that
	% \begin{equation*}\textstyle
	% 	\mathcal{W}_1(Q^q,P)\le \sup_{X \in \Lambda}(E_{Q^q}[X]-E[X]) + \varepsilon.
	% \end{equation*}
	% Dropping $\varepsilon$ yields the result, due to Proposition \ref{eq:transport-bsde}.
	As in the proof of Proposition \ref{prop:transp-bsde}, it follows from Corollary \eqref{cor:ForLipschitz} that $\lambda(E_{Q^q}[ \phi(W)] - E[\phi(W)]) - l(\lambda)\le E_{Q^q}[\int_0^Tg^*_u(q_u)\,du]$ for all $1$-Lipschitz function $\phi$ (i.e. satisfying \eqref{eq:1 Lips on C}).
	Thus, taking the supremum over such functions $\phi$ and over $\lambda\ge 0$ and using that $W$ is the canonical process yields
	\begin{equation*}\textstyle
		l^*\left(\sup_{\phi}(E_{Q^q}[\phi] - E[\phi])\right) \le E_{Q^q}\left[\int_0^Tg^*_u(q_u)\,du\right].
	\end{equation*}
	By the Kantorovich-Rubinstein formula, the above inequality implies \eqref{eq:transport-bsde_cano}.
\end{proof}

In particular, when $g(z)=z^2/2$, one recovers with \eqref{eq:transport-bsde_cano} (admittedly with less direct arguments) the Gaussian transport-entropy inequality on $C([0,1],\mathbb{R}^d)$ obtained by \citet{Fey-Ust04} and \citet{Lehec}.
Let us give a PDE characterization of the transport-type inequality derived above.
We use the notation
\begin{equation*}\textstyle
	\mathcal{W}_{1,d}(\gamma^{d}, \mu) := \inf\iint|x-y|\,d\pi(x,y) 
\end{equation*}
where the infimum is over the set of probability measures $\pi$ on $\mathbb{R}^{d }\times \mathbb{R}^{d}$ with marginals $\gamma^{d}$ and $\mu$.
\begin{corollary}
	Let $g:\Omega \times [0,T]\times \mathbb{R}^d\to [0,\infty)$ be convex, smooth and with bounded derivatives, and $l:\mathbb{R}\to [0,\infty)$ be an increasing convex function.
	Let $(s,x) \in[0,T]\times \mathbb{R}^d$, put $W^{s,x}_t:= x + W_t-W_s$, the Brownian motion started at $s$ with the value $x$.
	Consider the following claims:
	\begin{itemize}
		\item[(i)] ${\cal E}(\lambda f(W^{s,x}_T)) \le \lambda E[f(W^{s,x}_T)] + l(\lambda)$ for all $\lambda\ge0$ and $f:\mathbb{R}^d\to \mathbb{R}$ 1-Lipschitz
		\item[(ii)] $l^*\Big(\sqrt{T-s}\mathcal{W}_{1,d}(\mu_q,\gamma^d)\Big) \le \alpha(q)$ for all $q\in {\cal L}$ and $\mu_q =Q^q\circ (W^{s,x}_T/\sqrt{T-s})^{-1}$
		\item[(iii)] %Put $u(s,x):=E[f(W_T^{s,x})]$. 
		for every $\lambda \ge 0$, the solution $v^\lambda$ of the PDE
		\begin{equation}\textstyle
		\label{eq:PDE-charac}
			\begin{cases}
				&\partial_tv^\lambda + \frac{1}{2}\partial_{xx}v^\lambda + g_t( \partial_xv^\lambda) =0\\
				&v^\lambda(T,x) =  -l(\lambda)
			\end{cases}
		\end{equation}
		satisfies $v^\lambda(s,x)\le0$.
		\end{itemize}
	One has (i) is equivalent to (ii), and (i) or (ii) implies (iii).
\end{corollary}
\begin{proof}
	It follows as in the proof of Proposition \ref{prop:transp-bsde} that (i) is equivalent to 
	\begin{equation*}\textstyle
		l^*\Big(\sup_{f}\left(E_Q[f(W^{s,x}_T)] - E_P[f(W^{s,x}_T)] \right)\Big) \le \alpha(q)
	\end{equation*}
	where the supremum runs over the set of 1-Lipschitz functions.
	Thus, 
	\begin{equation*}\textstyle
		l^*\Big(\sqrt{T-s}\sup_{f}\Big(\int f\,d\mu_q - \int f\,d\gamma_d \Big) \Big) \le \alpha(q)
	\end{equation*}
	with $\mu_q = Q^q\circ(W^{s,x}_T/\sqrt{T-s})^{-1}$.
	This shows that (i) and (ii) are equivalent by the Kantorovich-Rubinstein formula.

	If (i) holds, then the function $v^\lambda(t,y):={\cal E}\Big(\lambda f(W^{t,y}_T)-\lambda E[f(W^{t,y}_T)]-l(\lambda)\Big)$ satisfies $v^\lambda(s,x)\le0$.
	Since $g$ is Lipschitz continuous, the minimal supersolution of the BSDE with generator $g$ coincides with its solution.
	By \cite[Proposition 4.4]{karoui01}, $v^\lambda$ solves the PDE \eqref{eq:PDE-charac}.
\end{proof}

\paragraph{3.2 Deviation inequalities.}
It is well-known that transport inequalities describe concentration properties and deviation inequalities.
Let us adapt the classical Marton's  argument  in our setting to derive deviation inequalities for  $X \in \Lambda$ (see \cite{Marton86}).
Below, denote by $f^{-1}$ the left-inverse of the increasing function $f$ and by $m_X$ a median of the random variable $X$, i.e. a real number such that $P(X \le m_X)\ge 1/2$ and $P(X \ge m_X)\ge 1/2$.
\begin{proposition}
	Let $g:\Omega \times [0,T]\times \mathbb{R}^d\to [0,\infty)$ be a convex function and $l:\mathbb{R}\to [0,\infty)$ be an increasing convex function such that
	\begin{equation}\textstyle
	\label{eq:transp-bsde-ineq}
				l^*(\mathcal{W}_1'( Q^q, P))\le E_{Q^q}\left[\int_0^Tg^*_u(q_u)\,du\right]\quad \text{for all } q \in  {\cal L}.
	\end{equation}
	If the PDE
	\begin{equation}\textstyle
	\label{eq:pde}
		\sup_{q \in \mathbb{R}^d}\left( g^*_t(q) -\partial_tf_t(x)- \frac{1}{2}\partial_{xx}f_t(x)|q|^2x^2- \partial_xf_t(x)|q|^2x \right) = f_0(1)/T
	\end{equation}
	admits a subsolution $\varphi:[0,T]\times(0,+\infty)\to\mathbb{R}$ which is increasing on $[\frac{1}{P(X \le m_X)},+\infty)$ for each $t \in [0,T]$, then for all $X \in \Lambda'$ one has
	\begin{equation}\textstyle
	\label{eq:deviation}
		P(X >m_X + r) \le \frac{1}{\varphi_T^{-1}\left(l^*\big(r - (l^*)^{-1}(\varphi_T(2))\big)\right)} \quad \text{for all } r> 0,
	\end{equation}
	with the convention $\frac{1}{0}:=+\infty$.
	The result also holds if $\mathcal{W}_1'$ is replaced by $\mathcal{W}_1$ and $\Lambda'$ by $\Lambda$.
\end{proposition}
\begin{proof}
	Let $r> 0$ be fixed, put $A:=\{X\le m_X\}$ and $B:= \Omega \setminus \{X\le m_X + r\}$.
	If $P(B)=0$, then the result is clear.
	Assume $P(B)>0$ and denote by $P^A$ and $P^B$ the probability measures absolutely continuous w.r.t. $P$ and with density $1_A/P(A)$ and $1_B/P(B)$, respectively.
	Notice that $\mathcal{W}'_1$ defined in Proposition \ref{prop:transp-bsde} satisfies the triangle inequality $\mathcal{W}_1'(P^A, P^B)\le \mathcal{W}_1'(P^A,P) + \mathcal{W}_1'(P,P^B)$.
	Taking $X \in \Lambda'$, one has
	\begin{equation*}\textstyle
		\mathcal{W}_1'(P^A, P^B)\ge E[-1_AX/P(A) + 1_BX/P(B)] \ge -m_X + (r + m_X) = r.
	\end{equation*}
	Thus, letting $q^A$ and $q^B$ in ${\cal L}$ be such that $M^{q^A}_T=1_A/P(A)$ and $M^{q^B}_T=1_B/P(B)$, respectively, it follows from \eqref{eq:transp-bsde-ineq} and the triangle inequality that
	\begin{equation*}\textstyle
		r \le (l^*)^{-1}\left(\alpha(q^A) \right) + (l^*)^{-1}\left(\alpha(q^B) \right).
	\end{equation*}
	From It\^o's formula, for every $q \in {\cal L}$ one has
	\begin{equation*}\textstyle
		E_{Q^q}[\varphi_T(M^q_T)] -\varphi_0(1) = E_{Q^q}\left[\int_0^T \partial_t\varphi_t(M^q_t)+\frac{1}{2}\partial_{xx}\varphi_t(M^q_t)|q_t|^2(M^q_t)^2 + \partial_x\varphi_t(M^q_t)|q_t|^2M^q_t\,dt \right].
	\end{equation*}
	Since $\varphi$ is a  subsolution of \eqref{eq:pde}, it holds $\alpha(q) \le E_{Q^q}[\varphi_T(M^q_T)]$.	
	Therefore, 
	\begin{align*}
		r &\le (l^*)^{-1}\left(E\left[\frac{1_A}{P(A)}\varphi_T(1_A/P(A)) \right] \right) + (l^*)^{-1}\left(E\left[\frac{1_B}{P(B)}\varphi_T(1_B/P(B)) \right] \right)\\
		& =  (l^*)^{-1}\left( \varphi_T(1/P(A)) \right) + (l^*)^{-1}\left(\varphi_T(1/P(B)) \right).
	\end{align*}
	Hence, using $P(A)\ge 1/2$, it follows
	\begin{equation*}
		P(B) = 1 - P(X\le m_X + r) \le \frac{1}{\varphi_T^{-1}\left(l^*\big(r - (l^*)^{-1}(\varphi_T(2))\big)\right)},
	\end{equation*}
	this concludes the proof.
\end{proof}

The most classical example arises when $g(z)=|z|^2/2$. 
In this case, ${\cal E}(X) = \log(E[e^X])$, by Proposition \ref{prop:transp-bsde} every $q \in {\cal L}$ satisfies \eqref{eq:transp-bsde-ineq}, and $\varphi(x)=\log(x)$ is a classical solution of \eqref{eq:pde}.
Thus the inequality \eqref{eq:deviation} becomes 
	\begin{equation}
	\label{eq:gaussian}
		P\left( X>r +m_X \right) \le e^{-l^*(r-(l^*)^{-1}(\log(2)))}\quad \text{for all }  r>0,
	\end{equation}
the classical  Gaussian concentration.
Notice that $\varphi$ is a subsolution of \eqref{eq:pde} if and only if
\begin{equation*}
	g^*_t(q) \leq \varphi_0(1)/T + \partial_t\varphi_t(x) + \left(\frac{1}{2}\partial_{xx}\varphi_t(x)x^2 + \partial_x\varphi_t(x)x \right)|q|^2 
\end{equation*}
for every $x >0$ and $q \in \mathbb{R}^d$.
Assume that $g$ is such that $g^*$ satisfies $g^*(q)\le B + C|q|^2$ for some $C>0$ and let $\varphi(x)=\frac{2}{k}\log(x)$.
It holds $\frac 12x^2\varphi''(x) + x\varphi'(x)=\frac 1k$.
Hence, if $\frac 1k> C$, then \eqref{eq:pde} holds.
It is important, however, to remember that the function $l$ such that \eqref{eq:transp-bsde-ineq} is satisfied is quadratic and must depend on $C$.

\paragraph{3.3 Dimension-free bounds.} For identically distributed claims, 
% Under positive homogeneity 
we show that the bounds obtained above become dimension-free:
\begin{proposition}
	Let $g:\Omega\times [0,T]\times\mathbb{R}\times \mathbb{R}^d\to [0,\infty)$ be a convex function (in $(y,z)$) increasing in $y$.
	If $(X_i)_{i=1,\dots,n}$ is a family of identically distributed random variables such that ${\cal E}(\lambda X_i) \le E[\lambda X_i]+l(\lambda) $ for all $\lambda\ge0$ and some convex increasing function $l:\mathbb{R}\to [0,\infty)$, then
	\begin{equation}\textstyle
	\label{eq:dim-free}
		{\cal E}\left(\lambda \frac{1}{n}\sum_{i=1}^nX_i\right) \le \lambda E[X_i] + l(\lambda) \quad \text{for all } \lambda\ge0.
	\end{equation}
\end{proposition}
\begin{proof}
	Since $g$ is convex, the functional ${\cal E}$ is also convex.
	Thus, one has 
	\begin{equation*}\textstyle
		{\cal E}\left(\lambda \frac{1}{n}\sum_{i=1}^nX_i\right) \le \sum_{i=1}^n\frac{1}{n}{\cal E}\left(\lambda X_i\right)\le \sum_{i=1}^n\frac{1}{n}\left( \lambda E[X_i] + l\left(\lambda\right)\right).
	\end{equation*}
The result follows since $X_i$ are identically distributed.
\end{proof}
The concentration inequality \eqref{eq:dim-free} is dimension-free in the sense that it does not depend on $n$, which implies that the bound cannot be improved by increasing the dimension.

In particular, if ${\cal E}(\frac{1}{n}\sum_{i=1}^nX_i)$ is viewed as the minimal superhedging price of the portfolio $\frac{1}{n}\sum_{i=1}^nX_i$, then the inequality \eqref{eq:dim-free} means that the superhedging price remains bounded from above when the portfolio increases and the bound is independent of the size of the portfolio.

\bibliographystyle{abbrvnat}
% \bibliography{references-Concen_RM}

\vspace{1cm}

\noindent Ludovic Tangpi, Fakult\"at f\"ur Mathematik, Universit\"at Wien; Austria\\ 
{\small\textit{E-mail address:} ludovic.tangpi@univie.ac.at}\\
Financial support from Vienna Science and Technology Fund (WWTF) under Grant MA 14-008 is gratefully acknowledged.
\end{document}